\documentclass[11pt,a4paper]{article}
\usepackage{amsmath}
\usepackage[T1]{fontenc}
\usepackage[utf8]{inputenc}
\usepackage{amsfonts}
\usepackage{amsthm}
\usepackage{amstext}
\usepackage{amssymb}
\usepackage{bbold}
\usepackage{hyperref}
\usepackage{url}

\newcommand{\zz}{\mathbb{Z}}
\newcommand{\rr}{\mathbb{R}}
\newcommand{\cc}{\mathbb{C}}

\newtheorem{theorem}{Theorem}

\newtheorem{lemma}[theorem]{Lemma}
\newtheorem*{lemma*}{Lemma}
\newtheorem{proposition}[theorem]{Proposition}
\newtheorem{corollary}[theorem]{Corollary}
\newtheorem{remark}[theorem]{Remark}

\newtheorem*{open*}{Open~question}

\begin{document}

\title{Root Separation for Trinomials}

\author{Pascal Koiran\\
LIP\thanks{UMR 5668 ENS Lyon, CNRS, UCBL.
The author is supported by ANR project
CompA (code ANR--13--BS02--0001--01).
Email:  
{\tt Pascal.Koiran@ens-lyon.fr}},
Ecole Normale Sup\'erieure de Lyon, Universit\'e de Lyon.
}

\maketitle

\begin{abstract}
  We give a  separation bound for the 
  complex roots of a trinomial $f \in \zz[X]$. The logarithm of the inverse of our separation bound is polynomial in the size of the sparse encoding of $f$; in particular, it is polynomial in $\log (\deg f)$.   It is known that no such bound is possible for 4-nomials (polynomials with 4 monomials).
  For trinomials, the classical results (which are based on the degree of $f$ rather than
  the number of monomials) give separation bounds that are exponentially worse.

As an algorithmic application, we show that the number of real roots of a trinomial $f$ can be computed in time polynomial in the size of the sparse encoding of~$f$. The same problem is open for 4-nomials.
\end{abstract}

\section{Introduction}

Root separation is
a classical topic in the study of polynomials in a real or complex variable.
As a motivation, note  for instance that a lower bound on the separation between the real roots of a polynomial provides information on the accuracy with which each real root should be approximated in order to  obtain isolating intervals.
The following inequality due to Mahler~\cite{mahler64} (see also~\cite{mignotte82}) is an example of a well-known result on root separation.
More recent results can be found e.g. in~\cite{bugeaud16,bugeaud10}.
\begin{theorem}[Mahler] \label{mahler}
  Let $f \in \zz[X]$ be a squarefree (i.e., without multiple roots)
   polynomial of degree $d \geq 2$.
  If $x_1,x_2$ are two distinct complex roots of  $f$ we have
  $$|x_1 - x_2| > \sqrt{3}(d+1)^{-(2d+1)/2} H^{-d+1}$$
  where $H$ denotes the maximum of the absolute values of the coefficients
  of $f$.
\end{theorem}
The first result of this paper is a separation bound for the real roots
of a trinomial, which we later extend to complex roots.
\begin{theorem} \label{main}
Consider a nonzero trinomial $f(x)=ax^{\alpha}+bx^{\beta}+cx^{\gamma} \in \zz[X]$
where $\alpha < \beta < \gamma$ and 
$\log \max(|a|,|b|,|c|,\alpha,\beta,\gamma) \leq s$ for some $s \geq 1$.
For any two distinct real roots $x_1,x_2$ of $f$ we have  $|x_1 - x_2| \geq \exp(-Cs^3)$ where $C$ is an absolute constant.
\end{theorem}
For a polynomial $f (X)= \sum_{i=1}^k a_i X^{\alpha_i} \in \zz [X]$ with $k$ nonzero coefficients, let us define 
the sparse size of $f$ as the sum of the bit sizes\footnote{Throughout the paper, $\log$ will denote the natural logarithm.} of the coefficients and exponents of $f$, i.e., as
\begin{equation} \label{size}
\sum_{i=1}^k [\log(1+|a_i|)+\log(1+\alpha_i)].
\end{equation}
We sometimes use the name "size of the sparse encoding of $f$" for the same quantity.
For $K=\rr$ or $K=\cc$, we denote  by $\sigma_{K}(f)$ the separation of the real (respectively, complex) roots of $f$, i.e., the minimum distance between two distinct roots of $f$ in $K$; 
and we denote the inverse of this number by  $\sigma_{K}^{-1}(f)$.
Theorem~\ref{main} therefore shows that the real roots of  trinomials are well separated, i.e., $\log(\sigma_{\rr}^{-1}(f))$ is polynomially bounded in the sparse size of $f$.
This is already false for polynomials with 4 terms: Michael Sagraloff~\cite{Sagraloff14} has given a family of examples showing that for 4-nomials,  $\log \sigma_{\rr}^{-1}(f)$  can be exponential in the sparse size of $f$.
Note also that in Theorem~\ref{mahler}, $\log \sigma_{\cc}^{-1}(f)$ is polynomial (in fact, almost linear) in
$\deg f$ rather than in $\log (\deg f)$,
and this is unavoidable in the dense case~\cite{mignotte82}.

As an algorithmic application of Theorem~\ref{main} we show that the number of real roots of a trinomial $f$ can be computed in time polynomial in the sparse size of $f$. Our model of computation is the standard Turing machine model.
An algorithm for the same problem was  already proposed in~\cite{rojas05} but their computational model is different: instead of the Turing machine model, they use an arithmetic model (i.e., they count
the number of arithmetic operations of their algorithm rather than the number of bit-level operations).

The proof of Theorem~\ref{main} is fairly short but somewhat 
 technical.
 Here is an  overview of the argument.
It follows from Descartes' rule of signs that a polynomial with $t$ monomials has at most $t-1$ positive real roots.
For a trinomial we have at most 2 positive real roots, and a grand total of 5 real roots at most.
The hardest task is to separate two roots of the same sign. Assume for instance that $0 < x_1 < x_2$.
By Rolle's theorem there must exist a point $m$ between the two roots where $f'(m)=0$.
Dividing by $x^{\alpha}$ if necessary, we can assume that $f$ has a nonzero constant term.
Hence $f'$ is a binomial and we can compute explicitly $m$ and $f(m)$.
The lower bound on $|x_1 - x_2|$ then follows from two ingredients:
\begin{itemize}
\item[(i)] We use Baker's theorem on linear forms in logarithms to show that $|f(m)|$ cannot be ``too small.''
\item[(ii)] We show that $|f'|$ cannot be ``too large'' on $[0,m]$.
\end{itemize}
A lower bound on $m-x_1$ follows from (i) and (ii), and this is in turn a lower bound on $x_2-x_1$.
Step (ii) is elementary but can only be achieved under an additional assumption on the exponents $\beta$ and $\gamma$ such as e.g. $\gamma \geq 2\beta$. Fortunately, if this condition is not satisfied we can work with the reciprocal polynomial $f^R(x)=x^{\gamma}f(1/x)$ instead.

In the computer science literature, Baker's theorem was  used by Etessami et al.~\cite{Etessami14} 
to compare ``succinctly represented integers.'' The same result\footnote{Etessami et al. then give an application to maximum probability parsing, a topic that is not considered in~\cite{Bastani11} or in the present paper.} was published a little earlier by Bastani et al.~\cite{Bastani11} with essentially the same
proof (instead of Baker's theorem, they rely on a variation by Nesterenko).
We use Baker's theorem in essentially the same way as~\cite{Bastani11} and~\cite{Etessami14}, but instead of a comparison algorithm we derive a separation bound.

The role of Rolle's theorem in the proof of Theorem~\ref{main} is
the main obstacle toward a separation result for complex roots.
There is fortunately a significant body of work on adaptations of
Rolle's theorem to the complex domain (see e.g.~\cite{marden85}).
We use the Grace-Heawood theorem and a variation due to Marden~\cite{marden66}.

\subsection*{Organization of the paper}

We begin in Section~\ref{realsec}
with a proof of our separation bound for real roots
(Theorem~\ref{main}) because it is simpler than the corresponding
proof 
for the complex domain, and it provides the ingredients needed
for our algorithmic application.
This application (counting the number of real roots
of a trinomial in polynomial time) appears in Section~\ref{algosec}.
In Section~\ref{complexsec} we extend Theorem~\ref{main} to complex roots,
first for squarefree trinomials and then in Theorem~\ref{sepc}
for all trinomials.
We conclude in Section~\ref{final} with two open problems.
Throughout the paper, $C_1,C_2,\ldots...$ denote suitable 
positive constants.

\section{Separation of Real Roots} \label{realsec}

In this section we give a proof of Theorem~\ref{main}.
Baker's theorem on linear forms in logarithms is our main tool.
More precisely, we use the following version
by Baker and  W{\"u}stholz~\cite{baker93}.
\begin{theorem} \label{baker}
  Given two real numbers $s,t \geq 1$,
  let $a=(a_1,\ldots,a_n)$ be a list of $n$ positive rational numbers such that $a_i=p_i/q_i$ where the positive integers $p_i,q_i$ satisfy
  $\log \max(p_i,q_i) \leq s$ for all $i$;
let  $b=(b_1,\ldots,b_n)$ be a list of $n$ arbitrary integers such that $\log |b_i| \leq t$ for all $i$.
Let $$\Lambda(a,b)=b_1 \log a_1 + b_2 \log a_2 + \cdots + b_n \log a_n.$$

  If $\Lambda(a,b) \neq 0$ then $|\Lambda(a,b)| \geq \exp[-C(n)ts^n]$,
  where $$C(n)=18(n+1)!n^{n+1}32^{n+2}\log(2n).$$
\end{theorem}
The theorem established in~\cite{baker93} is in fact more general:
the $a_i$'s are allowed to be algebraic numbers rather
than just integers.
The version stated in~\cite{Etessami14} is {\em less} general
than Theorem~\ref{baker} since it considers only the case when the $a_i$'s
are integers.
That paper also provides a detailed discussion, geared towards the nonexpert
in number theory, of the translation from the general statement
in~\cite{baker93} (where the role of the absolute value of integers is played by the Weil height of algebraic numbers) to a more elementary setting.

We will need to apply Theorem~\ref{baker} for $n = 2$ only.
Note that a linear form in 2 logarithms of rational numbers
can obviously be rewritten as a linear form in 4 logarithms of integers.
However, this would result in a severe degradation of our bounds due
to the presence  
of the factor $s^n$ in the lower bound for
$|\Lambda(a,b)|$ provided by Theorem~\ref{baker}.

\begin{proposition} \label{binomial}
  Let $g(x)=a_1x^{\beta} +a_2$ where  $a_1,a_2$ are two nonzero integers and $\beta$ is a positive rational number. Write $\beta=\beta_1 / \beta_2$, where $\beta_1$ and $\beta_2$ are positive integers.
  We evaluate the binomial $g$ at a rational point $p/q$ where
  $p$, $q$ are positive integers.
  
  Assume that $\log \max(p,q,|a_1|,|a_2|) \leq s$
  and $\log \max (\beta_1,\beta_2) \leq t$ where $s,t \geq 1$.
If $g(p/q) \neq 0$ then $|g(p/q)| \geq \exp(-C_1 ts^2)$.
\end{proposition}
\begin{proof}
  We can assume that $a_1$ and $a_2$ are of opposite signs: otherwise,
  $|g(p/q)| \geq |a_2| \geq 1$.
  Assume for instance that $a_2 < 0$ and $a_1 > 0$.

  In order to compare $a_1(p/q)^{\beta}$ to $|a_2|$, we compare the logarithms
  of these numbers.
  Let $$\delta = \log[a_1 (p/q)^{\beta}] - \log |a_2|=
  [\beta_2 \log (a_1/|a_2|) + \beta_1 \log (p/q)]/\beta_2.$$
  By Theorem~\ref{baker}, the absolute value of the numerator is lower bounded
  by $\exp[-C(2)ts^2]$.
  Since $\beta_2 \leq e^t$ we have $|\delta| \geq \exp[-C(2)ts^2-t]$.
  Consider first the case $\delta > 0$.
  We have $\log[a_1 (p/q)^{\beta}] = \log |a_2| + \delta$, therefore
  $a_1 (p/q)^{\beta} \geq |a_2| e^{\delta} \geq |a_2| (1+\delta)$
  and  $a_1 (p/q)^{\beta} - |a_2| \geq |a_2|\delta \geq \exp[-C(2)ts^2-t]$.
  We now turn our attention to the case $\delta < 0$, and distinguish
  two subcases.
  \begin{itemize}
  \item[(i)] If $\delta \leq -1$ we have $a_1 (p/q)^{\beta} \leq |a_2|/e$,
    hence $|g(p,q)| \geq |a_2| (1-1/e) \geq 1-1/e$.
  \item[(ii)] If $-1 \leq \delta \leq 0$ then by convexity of the exponential function:
    $$a_1 (p/q)^{\beta} = |a_2| e^{\delta} \leq |a_2|.[1+\delta(1-1/e)].$$
   Hence   $a_1 (p/q)^{\beta} - |a_2| \leq \delta(1-1/e) |a_2| \leq - (1-1/e)\exp[-C(2)ts^2-t].$
  \end{itemize}
  In all cases, $|g(p/q)| \geq \exp(-C_1ts^2)$ for some suitable constant $C_1>C(2)$.
\end{proof}

\begin{lemma} \label{fofm}
  Consider a trinomial $f(x)=a+bx^{\beta}+cx^{\gamma} \in \zz[X]$
  where $\gamma > \beta$. If $f'(m)=0$ at some point $m>0$
  then $m= |b\beta / c\gamma|^{1/(\gamma-\beta)}$ and
    $$f(m)=a+b\left(1-\frac{\beta}{\gamma}\right)m^{\beta}=a+b\left(1-\frac{\beta}{\gamma}\right) \left| \frac{b\beta}{c\gamma} \right|^{\beta/(\gamma-\beta)}.$$
\end{lemma}
This follows from a routine calculation. Note that $b$ and $c$ must be nonzero.
\begin{lemma} \label{min}
Consider a trinomial $f(x)=a+bx^{\beta}+cx^{\gamma} \in \zz[X]$
where $\gamma > \beta$, $a\neq 0$ and
$\log \max(|a|,|b|,|c|,\beta,\gamma) \leq s$ for some $s \geq 1$.
If $f'(m)=0$ at some point $m>0$ where $f(m) \neq 0$ then $|f(m)| \geq \exp(-C_2s^3)$.
\end{lemma}
This follows from Lemma~\ref{fofm} and Proposition~\ref{binomial}
(which we apply with $t=s$).
Note that we need the hypothesis $a \neq 0$ in order to apply this proposition.
\begin{lemma} \label{f'}
  Consider a trinomial $f(x)=a+bx^{\beta}+cx^{\gamma} \in \zz[X]$ where
  $\gamma \geq 2\beta$. If $f'(m)=0$ at some point $m>0$ then
  $\sup_{x \in [0,m]} |f'(x)| \leq b^2 \beta$. 
\end{lemma}
\begin{proof}
  The supremum is reached at 0, or at a point $r \in ]0,m[$ where $f''(r)=0$.
    We are in the first case if (and only if)  $\beta=1$, and then $f'(0)=b$.
    In the second case, 
    \begin{equation} \label{f'r}
    f'(r) = b \beta \left(1-\frac{\beta-1}{\gamma-1}\right)\left|
    \frac{b\beta (\beta-1)}{c\gamma(\gamma-1)}\right|^{(\beta-1)/(\gamma-\beta)}.
    \end{equation}
    This follows from a direct computation,
    or an application of Lemma~\ref{fofm} to $f'(x)=b\beta x^{\beta-1}+c\gamma x^{\gamma-1}$ (which we can treat as a trinomial without a constant term).
    Taking into account the condition $\gamma \geq 2\beta$,
    we see that the exponent $(\beta-1)/(\gamma-\beta)$ in~(\ref{f'r}) is smaller than 1. The result then follows from  the crude upper bound: $\displaystyle \left|
    \frac{b\beta (\beta-1)}{c\gamma(\gamma-1)}\right| \leq |b|$. 
        \end{proof}

\begin{remark}
In the previous lemma, we can  bound $|f'|$ without using the second derivative of $f$. Since $f'$ has a positive root, $b$ and $c$ must be of oppositive signs.
Hence for $x \geq 0$,  $|f'(x)|=|b\beta x^{\beta-1}+c\gamma x^{\gamma-1}| \leq \max(|b|\beta x^{\beta-1}, |c|\gamma x^{\gamma-1})$
and   $\sup_{x \in [0,m]} |f'(x)| \leq \max(|b|\beta m^{\beta-1},| c|\gamma m^{\gamma-1}) =
 |b| \cdot \beta  \cdot |b\beta / c\gamma|^{{(\beta-1)}/(\gamma-\beta)} \leq b^2 \beta$.
 This is the same bound as in the statement of Lemma~\ref{f'}, but slightly worse than~(\ref{f'r}). 
 One advantage of the proof of Lemma~\ref{f'} is that~(\ref{f'r}) shows that some relation between $\beta$ and $\gamma$ (such as $\gamma \geq 2\beta$) is {\em necessary} in order to obtain a good control on $|f'|$.
 \end{remark}
 
 Before proving our root separation result, we recall Cauchy's bound~\cite{cauchy1829}. A number of similar bounds can be found in~\cite{rahman2002}, chapter~8.
 \begin{proposition} \label{cauchy}
 Let $P (X) = \sum_{ i = 0}^d a_i X^i$ be a complex polynomial of degree $d$. For any root $x$ of $P$, we have 
 $$|x| \leq 1+\max\left(\frac{|a_0|}{|a_d|},\frac{|a_1|}{|a_d|},\ldots,\frac{|a_{d-1}|}{|a_d|}\right).$$
 \end{proposition}
 
 \begin{corollary} \label{smallroots}
  Let $P (X) = \sum_{ i = 0}^d a_i X^i$ be a complex polynomial with $a_0 \neq 0$.  For any root $x$ of $P$, we have 
 $$|x| \geq \frac{1}{ 1+\max\left(\frac{|a_d|}{|a_0|},\frac{|a_{d-1}|}{|a_0|},\ldots,\frac{|a_{1}|}{|a_0|}\right)}.$$
 
 Moreover, for a polynomial $f(X)=\sum_{i=0}^d f_i X^i \in \zz[X]$ with $|f_i| \leq M$ for all $i$, any nonzero root 
 $x$ of $f$ satisfies $|x| \geq 1/(M+1)$.
 \end{corollary}
 The first part of this corollary follows from an application of  Proposition~\ref{cauchy} to the reciprocal polynomial $P^R(X)=X^d P(1/X)$. The second part follows from the first by factoring out the highest power of $x$ which divides $f$.
 We can now give the proof of our main result.

\begin{proof}[Proof of Theorem~\ref{main}]
We first consider the case where one of the two roots, e.g. $x_1$, is equal to $0$. By the second part of Corollary~\ref{smallroots} we have $|x_1-x_2| \geq 1/(1+e^s)$.
If the two roots are nonzero and of oppositive signs, for the same reason we have  $|x_1-x_2| \geq 2/(1+e^s)$.
It therefore remains to consider the case where the two roots are nonzero and of the same sign. 

Assume for instance
that $0 < x_1 < x_2$. 
The first coefficient $a$ must be nonzero, otherwise $f$ is a binomial and can have at most one positive root. Also, we may (and will) assume that $\alpha=0$ by factoring out $x^{\alpha}$ if necessary.
By Rolle's theorem, there is a point $m \in ]x_1,x_2[$ where $f'(m)=0$. Since $f$ is a trinomial it can have at most two positive roots, hence $f(m) \neq 0$. We can therefore apply Lemma~\ref{min} and conclude that $|f(m)| \geq \exp(-C_2s^3)$.
Such a lower bound is useful because $$|x_2 - x_1| \geq |m-x_1| \geq \frac{|f(m)-f(x_1)| }{\sup_{x\in [x_1,m]} |f'(x)| } \geq  
\frac{|f(m)|}{ \sup_{x \in [0,m] }|f'(x)|}.$$
At this point, we need to distinguish two further cases.
\begin{itemize}
\item[(i)] If $\gamma \geq 2\beta$ we can apply Lemma~\ref{f'} and conclude that $$|x_2 - x_1| \geq \frac{\exp(-C_2s^3)}{b^2\beta} \geq \exp(-C_2s^3-3s).$$
\item[(ii)] If $\gamma < 2\beta$ we consider the reciprocal polynomial $f^R(x)=x^{\gamma}f(1/x)=c+bx^{\gamma-\beta}+ax^{\gamma}.$ 
In $f^R$ the ratio of the two highest exponents is $$\frac{\gamma }{\gamma - \beta} = \frac{1 }{1-\beta/\gamma} 
\geq  2.$$
It follows that $f^R$ falls within the scope of case (i); since $1/x_1$ and $1/x_2$ are roots of $f^R$, the lower bound 
from case~(i) applies to $|\frac{1}{x_1}-\frac{1}{x_2}|.$
This yields a lower bound on $|x_2-x_1|$ since 
$$|x_2-x_1|=x_1x_2 |\frac{1}{x_1}-\frac{1}{x_2}|$$
and we have already obtained the lower bound $x_2 \geq x_1 \geq 1/(1+e^s)$.
\end{itemize}
We conclude that $|x_1 - x_2| \geq \exp(-Cs^3)$ in all cases, for some appropriate constant $C > C_2$. 
\end{proof}

\section{Computing the Number of Real Roots} \label{algosec}

We now give an algorithmic application of Theorem~\ref{main}.
\begin{theorem} \label{isolation}
The number of real roots of a trinomial  $f(x)=ax^{\alpha}+bx^{\beta}+cx^{\gamma} \in \zz[X]$ can be computed in time polynomial in the sparse size of $f$.
\end{theorem}
\begin{proof}
We will assume that the 3 coefficients $a,b,c$ are nonzero: if one coefficient vanishes $f$ is a binomial and we leave this easy case to the reader.
It is also easy to reduce to the counting of positive roots: $f$ has 0 as a root if and only if $\alpha>0$, and the number
of negative roots of $f$ is the number of positive roots of $f(-x)$.
It therefore remains to decide whether the input trinomial $f$ has 0, 1 or 2 positive roots.
Dividing by $x^{\alpha}$ if necessary, we will further assume that $\alpha=0$.
If $a$ and $c$ are of opposite signs, $f$ has exactly one positive root: it must have a root since $f$ takes opposite signs
at 0 and $+\infty$, and for the same reason we cannot have two roots.

When $a$ and $c$ are of the same sign, there are 0 or 2 positive roots and we determine which of these two values is the correct one using Lemma~\ref{fofm}. More precisely, assume for instance that $a$ and $c$ are positive.
We'll have 2 positive roots if and only if there exists $m>0$ such that $f'(m)=0$ and $f(m)<0$.
By Lemma~\ref{fofm}, we are in this case if and only if $b<0$ and 
   $$\frac{a\gamma}{|b|(\gamma-\beta)} < \left| \frac{b\beta}{c\gamma} \right|^{\beta/(\gamma-\beta)}.$$
This is equivalent to

$$  \left| \frac{a\gamma}{b(\gamma-\beta)}\right|^{\gamma-\beta} < \left| \frac{b\beta}{c\gamma} \right|^{\beta},$$
and therefore to:
\begin{equation} \label{discriminant}
  (a\gamma)^{\gamma-\beta}(c\gamma)^{\beta} < | b\beta|^{\beta} |b(\gamma-\beta)|^{\gamma-\beta}.
  \end{equation}
The latter inequality is an instance of the problem of comparing "succinctly represented integers"~\cite{Etessami14}.
As pointed out above, it is shown in~\cite{Etessami14} and~\cite{Bastani11}
how to solve this problem in polynomial time using Baker's theorem.
\end{proof}
As pointed out by an anonymous referee, checking~(\ref{discriminant}) amounts
to computing the sign of the so-called ``$A$-discriminant'' of the trinomial
$a+bx^{\beta}+cx^{\gamma}$ (Proposition~1.8 in~\cite{gelfand}, chapter 9).

The first version of this
paper\footnote{\href{https://arxiv.org/abs/1709.03294v1}{arxiv.org/abs/1709.03294v1}} claims that Theorem~\ref{isolation}
follows from  the algorithm of Jindal and Sagraloff~\cite{jindal17} for
the isolation of real roots of sparse polynomials. As pointed out by Michael Sagraloff (personal communication), this requires a more careful justification.
I wrote that the complexity of their algorithm is polynomial in the sparse size of $f$ and in $\log \sigma_{\rr}^{-1}(f)$.
This was a misstatement of their result: even if the real roots of $f$
are well separated, their algorithm may not run in polynomial time if
$f$ has two conjugate complex roots with very small imaginary parts.
In the full version of their paper~\cite{jindal18}, Jindal and Sagraloff
show that the conjugate complex roots of trinomials are well separated.
This implies that the real roots of trinomials can indeed be isolated in
polynomial time. In their paper, Jindal and Sagraloff provide explicit
constants for their separation bound. This is required for an actual
implementation of their algorithm, especially for the correct
treatment of double roots (see Section~9 of~\cite{jindal18} for details).
The fact that conjugate complex roots of trinomials are well separated
also follows from the results of the next section, where we show that
{\em all} complex roots of trinomials are well separated.

\section{Separation of Complex Roots} \label{complexsec}

In this section we show that the separation bound of Theorem~\ref{main}
also applies to complex roots of trinomials.
We begin with the the easy case of binomials.
In this case, the nonzero roots are uniformly distributed on a circle.
\begin{proposition} \label{binsep}
  Consider a binomial $f(x)=bx^{\beta}+cx^{\gamma} \in \zz[X]$ with $b,c \neq 0$ and $\beta < \gamma$.
  For any two complex roots $x_1,x_2$ of $f$ we have: $$|x_1 - x_2| \geq \frac{1}{|c|}\sqrt{2\left(1-\cos \frac{2\pi}{\gamma-\beta}\right)}.$$
  
\end{proposition}
\begin{proof}
  Set $\varepsilon=1$ if $b$ and $c$ are of opposite signs, $\varepsilon=-1$ if they are
  of the same sign.
  The roots of $f$ are the $\gamma-\beta$ numbers of the form
  $x=|b/c|^{1/(\gamma-\beta)} \xi$, where $\xi^{\gamma-\beta}=\varepsilon$.
  In addition, 0 is a root of multiplicity $\beta$ if $\beta \geq 1$.
  The distance between the origin and any other root is
  $|b/c|^{1/{(\gamma-\beta)}} \geq 1/|c|.$
  The distance between any two other roots is at least
  $$|1-e^{2i\pi/(\gamma-\beta)}|/|c| =
  \frac{1}{|c|}\sqrt{2\left(1-\cos \frac{2\pi}{\gamma-\beta}\right)}.$$
  \end{proof}

\subsection{Squarefree trinomials}

In order to extend Theorem~\ref{main} to the complex domain,
a suitable adapatation of Rolle's theorem is needed.
The following one is due to Grace and Heawood (see for instance Theorem~5.1
in Marden's survey~\cite{marden85}; for a proof, see Theorem~23.1 in his book~\cite{marden66} or Theorem~4.3.1 in~\cite{rahman2002}).
\begin{theorem} \label{rollec}
If $z_1$ and $z_2$ are two distinct roots of a polynomial $f \in \cc[X]$
  of degree $n$, then at least one zero of $f'$ lies in the disk of center
  $M=(z_1+z_2)/2$ and radius
  $$r=\frac{|z_1 - z_2|}{2}\cot(\pi/n).$$
\end{theorem}
Here is an analogue of Lemma~\ref{fofm}:
\begin{lemma} \label{fofmc}
  Consider a trinomial $f(x)=a+bx^{\beta}+cx^{\gamma} \in \zz[X]$
  where $\gamma > \beta$. If $f'(m)=0$ at some point $m \in \cc^*$
  then $m= |b\beta / c\gamma|^{1/(\gamma-\beta)} \xi$ where
$\xi^{\gamma-\beta}=\pm 1$.
Moreover,
$$f(m)=a+b\left(1-\frac{\beta}{\gamma}\right) m^{\beta} =
a+b\left(1-\frac{\beta}{\gamma}\right)\left| \frac{b\beta}{c\gamma} \right|^{\beta/(\gamma-\beta)}\xi^{\beta}.$$
\end{lemma}
The complex analogue of Lemma~\ref{min} is:
\begin{lemma} \label{minc}
Consider a trinomial $f(x)=a+bx^{\beta}+cx^{\gamma} \in \zz[X]$
where $\gamma > \beta$, $a\neq 0$ and
$\log \max(|a|,|b|,|c|,\beta,\gamma) \leq s$ for some $s \geq 1$.
If $f'(m)=0$ at some point $m \in \cc$ where $f(m) \neq 0$ then $|f(m)| \geq \exp(-C_3s^3)$.
\end{lemma}
\begin{proof}
  At $m=0$ we have $|f(m)| = |a| \geq 1$. It therefore remains to 
  study the case $m \neq 0$.
  Consider the complex number $\xi$ from Lemma~\ref{fofmc}.
  If $\xi^{\beta}$ is real then $\xi^{\beta} = \pm 1$ and the result follows
  from Proposition~\ref{binomial} just like in Lemma~\ref{min}.

  If $\xi^{\beta} {\not \in} \rr$,
  let us write $\xi^{\beta} = \cos \theta + i \sin \theta$ with $\sin \theta \neq 0$. Define also: $$r=b\left(1-\frac{\beta}{\gamma}\right)\left| \frac{b\beta}{c\gamma} \right|^{\beta/(\gamma-\beta)}.$$
  By Lemma~\ref{fofmc} we have $f(m)=a+r(\cos \theta + i \sin \theta)$ so that
  $$|f(m)|^2 = (a + r cos \theta)^2 + (r \sin \theta)^2=a^2+r^2+2ar\cos \theta.$$
  Consider now the case where $a$ and $r$ are of opposite signs.
  We have  $|f(m)|^2 = (a+r)^2 + 2ar(\cos \theta - 1) \ \geq (a+r)^2$.
  If $a+r \neq 0$ then $|a+r| \geq \exp(-C_2 s^3)$ by Lemma~\ref{min} and we are done.
  
  If $a+r=0$ then $|f(m)|^2 = 2a^2(1-\cos \theta) \geq 2(1-\cos \theta)$.
  Recall from Lemma~\ref{fofmc}
  that $\xi^{2(\gamma-\beta)} = 1$. Hence we may assume that $\theta = 2k\pi/n$
  where $n=2(\gamma-\beta)$ and $k \in \{1,\ldots,2n-1\}$ ($k=0$ is impossible
  due to the assumption $\sin \theta \neq 0$).
  To conclude, we note that $1 - \cos \theta \geq 1- \cos (2\pi/n)
  \geq 1 - \cos (\pi/\gamma)$.
  For large values of $\gamma$ we have
  $\cos (\pi/\gamma) = 1 - \pi^2/(2\gamma^2) +o(1/\gamma^2)$.
  Hence there is a constant $\gamma_0$ such that $\cos (\pi/\gamma)< 1 - 1/\gamma^2$ for all $\gamma \geq \gamma_0$.  Under this condition we have $1-\cos \theta \geq 1 /\gamma^2 \geq e^{-2s}$ and $|f(m)| \geq e^{-s}$. For $\gamma < \gamma_0$ we have $|f(m)|^2 \geq 2(1-\cos \theta) \geq 2[1 - \cos (\pi/\gamma)] > 2[1 - \cos (\pi/\gamma_0)]$. Since $\gamma_0$ is an absolute constant, there is a constant
  $c \geq 1$ such that the bound $|f(m)| \geq e^{-cs}$ applies to both cases
  $\gamma \geq \gamma_0$ and $\gamma < \gamma_0$.

  The case where $a$ and $r$ are of the same sign is handled by a similar argument, noting that this time we have:
  $$|f(m)|^2 = (a-r)^2 + 2ar(1 + \cos \theta ) \ \geq (a-r)^2.$$
\end{proof}
We are now ready to obtain a separation bound in the complex domain.
As a first step, we consider only the case where $f$ is squarefree.
\begin{proposition} \label{squarefree}
Consider a nonzero squarefree trinomial $f(x)=ax^{\alpha}+bx^{\beta}+cx^{\gamma} \in \zz[X]$
where $\alpha < \beta < \gamma$ and 
$\log \max(|a|,|b|,|c|,\alpha,\beta,\gamma) \leq s$ for some $s \geq 1$.
For any two distinct complex roots $x_1,x_2$ of $f$ we have  $|x_1 - x_2| \geq \exp(-C_4 s^3)$.
\end{proposition}
\begin{proof}
  If one of the two roots is equal to 0, the argument used in the proof
  of Theorem~\ref{main} (namely, Corollary~\ref{smallroots}) still applies.
  We will therefore assume that $x_1$ and $x_2$ are nonzero.
  Again, we assume that $\alpha=0$ by factoring out $x^{\alpha}$
  if necessary.\footnote{We already have $\alpha \leq 1$
    before the factorization since $f$
  is assumed to be squarefree.}
  Then we appeal to the complex version of Rolle's theorem:
  by Theorem~\ref{rollec} there is a point such that $f'(m)=0$
  in the disk of center
  $M=(x_1+x_2)/2$ and radius
  \begin{equation} \label{req}
   r=\frac{|x_1 - x_2|}{2}\cot(\pi/\gamma).
    \end{equation}
  Since $f$ is squarefree, $f(m) \neq 0$ so that $|f(m)| \geq \exp(-C_3s^3)$
  by Lemma~\ref{minc}.
  Let us denote by $I=[m,x_1]$ the line segment joining $m$ to $x_1$
  in the complex plane.
  Like in the real case,
  we have $|f(x_1)-f(m)| \leq |m-x_1|\sup_{z \in I}|f'(z)|$ so that
  \begin{equation} \label{accroissements}
|m - x_1| \geq\frac{ |f(m)|}{ \sup_{z \in I}|f'(z)|}.
\end{equation}
  This lower bound on $|m - x_1|$ is 
  useful since:
  $$\frac{|x_1 - x_2|}{2}\cot(\pi/\gamma) \geq |m-M| \geq |m-x_1|-|x_1-M|
  =|m-x_1|-|x_1-x_2|/2.$$
  The first inequality holds true because $m$ belongs to the circle
  of center $M$ and radius given by~(\ref{req}).
  As a result, we have
  \begin{equation} \label{triangle}
    |x_1-x_2| \geq \frac{2|m-x_1|}{1+\cot(\pi/\gamma)}.
    \end{equation}
  If $m=0$,  the lower bound on $|x_1|$ from Corollary~\ref{smallroots} recalled at the beginning of the proof applies to the the numerator of~(\ref{triangle}). As to the denominator,
  note that $\cot(\pi/\gamma) = \gamma / \pi+o(\gamma)$ for large $\gamma$.
  As a result there is a constant $A>0$ such that
  $1+\cot(\pi/\gamma) \leq A\gamma$ for all $\gamma \geq 1$.
 We therefore have 
 a lower bound of the form  $|x_1 - x_2| \geq \exp(-C_5s)$
 in the case $m=0$.
 This is  better than the lower bound claimed in the statement of
 Proposition~\ref{squarefree},
so we'll assume that $m \neq 0$ in the remainder of the proof.

  Our next goal is to bound $\sup_{z \in I}|f'(z)|$ in order to apply~(\ref{accroissements}) and~(\ref{triangle}).
  Note that for $z \in I$ we have $|z| \leq \max(|m|,|x_1|)$.
  We distinguish two cases:
  \begin{itemize}
  \item[(i)]   If $|x_1| \leq (1+1/\gamma)|m|$
    then $|z| \leq (1+1/\gamma)|m|$ for all $z \in I$.
We have $|f'(z)| \leq |b \beta||z|^{\beta-1}+|c\gamma||z|^{\gamma-1}.$
Since $(1+1/\gamma)^{\gamma} \leq e$ we have for any $z \in I$:
$$|f'(z)| \leq e(|b \beta||m|^{\beta-1}+|c\gamma||m|^{\gamma-1})=2e|b \beta||m|^{\beta-1}.$$
From the explicit  value
    $|m| = |b\beta / c\gamma|^{1/(\gamma-\beta)}$ computed in Lemma~\ref{fofmc}
we have
\begin{equation} \label{deriv}
|f'(z)| \leq 2e|b \beta||b\beta / c\gamma|^{(\beta-1)/(\gamma-\beta)}.
\end{equation}
At this point, like in the proof of Theorem~\ref{main}
we need to throw in the additional assumption $\gamma \geq 2 \beta$.
Under this assumption, from~(\ref{deriv}) we have $|f'(z)| \leq 2eb^2 \beta$.
Combining~(\ref{accroissements}) and~(\ref{triangle}) we have:
$$ |x_1-x_2| \geq \frac{\exp(-C_3s^3)}{eb^2 \beta(1+\cot(\pi/\gamma))},$$
which yields the desired bound.
If $\gamma < 2\beta$, we work instead with the reciprocal polynomial
$f^R$ instead of $f$ like in the proof of Theorem~\ref{main}.
Note that $f^R$ is squarefree since~$f$ is.

  \item[(ii)] If $|x_1| \geq (1+1/\gamma)|m|$ then $|m-x_1| \geq  |m|/\gamma$.
        Since $|m| = |b\beta / c\gamma|^{1/(\gamma-\beta)}$
        we have $|m| \geq \exp(-2s)$. 
    Inequality (\ref{triangle}) then yields a lower
    bound of the form $|x_1 - x_2| \geq \exp(-C_5s)$, and we are done.
    \end{itemize}
\end{proof}

\subsection{The general case}

In order to obtain our final result on root separation, it remains to get rid
of the assumption that $f$ is squarefree. For this, we shall use a version
of Theorem~\ref{rollec} that can take multiple zeros into account.
The following result is due to Marden~(\cite{marden66}, Theorem~25.1).

\begin{theorem} \label{marden}
 Let  $f \in \cc[X]$ be a polynomial of degree $n$.
 If $z_1$ and $z_2$ are two distinct roots of $f$ of respective multiplicities $k_1$ and $k_2$, then at least one zero $p$ (different from $z_1$ and $z_2$)
 of $f'$ lies in the disk of center
  $M=(z_1+z_2)/2$ and radius
  $$\rho=\frac{|z_1 - z_2|}{2}\cot(\pi/2q)$$
  where $q=n+1-k_1-k_2$.
\end{theorem}
For the present paper, the main advantage of Theorem~\ref{marden} over Theorem~\ref{rollec} is not so much the role played  by the multiplicities $k_1,k_2$ than the guarantee that $p$ is distinct from $z_1$ and $z_2$. As far as we know, there is no such guarantee in Theorem~\ref{rollec}.

\begin{theorem} \label{sepc}
Consider a nonzero trinomial $f(x)=ax^{\alpha}+bx^{\beta}+cx^{\gamma} \in \zz[X]$
where $\alpha < \beta < \gamma$ and 
$\log \max(|a|,|b|,|c|,\alpha,\beta,\gamma) \leq s$ for some $s \geq 1$.
For any two distinct complex roots $x_1,x_2$ of $f$ we have  $|x_1 - x_2| \geq \exp(-C_6 s^3)$.
\end{theorem}
\begin{proof}
  Like in the proofs of Theorem~\ref{main} and Proposition~\ref{sepc}
  we'll work without loss of generality under the assumptions
  $\alpha=0$ and $\gamma \geq 2\beta$. Also, we only have to treat
  the case where the roots $x_1$ and $x_1$ are nonzero.
  Let $m$ be the root of $f'$ introduced in the proof of Proposition~\ref{sepc}.
  In the case $f(m) \neq 0$, we have obtained in that proof the lower bound
  $|x_1 - x_2| \geq \exp(-C_4 s^3)$.
  If $f(m)=0$, $m$ is a double root of $f$ and we can apply Theorem~\ref{marden}
  with $z_1=x_1$, $z_2=m$, $k_1=1$ and $k_2=2$.
  Hence there exists a zero $p$ of $f'$, distinct from $m$ and $x_1$,
  in the disk of center $M'=(x_1+m)/2$ and radius
  $$\rho=\frac{|x_1 - m|}{2}\cot\frac{\pi}{2(\gamma-2)}.$$
 Since $m$ and $p$
  are two distinct roots of the binomial $f'$, by Proposition~\ref{binsep}
  we have
  \begin{equation} \label{mplower}
    |m-p| \geq
    \frac{1}{|c\gamma|}\sqrt{2\left(1-\cos \frac{2\pi}{\gamma-1}\right)}.
    \end{equation}

  After this lower bound on $|m-p|$, we shall give an upper bound.
  Recall that $m$ lies in the disk of center $M=(x_1+x_2)/2$ and radius $r$
  defined in~(\ref{req}).
  By the triangle inequality,
  \begin{equation} \label{x1m}
    |x_1-m| \leq r+ |x_1-x_2|/2 = \frac{|x_1-x_2|}{2}[1+\cot(\pi/\gamma)],
    \end{equation}
  hence
  \begin{equation} \label{rhobound}
    \rho\leq \frac{|x_1 - x_2|}{4}[1+\cot(\pi/\gamma)]\cot\frac{\pi}{2(\gamma-2)}.
    \end{equation}
  Another application of the triangle inequality shows that
  $$|m-p| \leq |m-M'|+\rho=|m-x_1|/2+\rho.$$
  From~(\ref{x1m}) and~(\ref{rhobound}) we have
  \begin{equation} \label{mpupper}
    |m-p| \leq \frac{|x_1 - x_2|}{4}[1+\cot(\pi/\gamma)][1+\cot\frac{\pi}{2(\gamma-2)}].
    \end{equation}
  To conclude, we just have to put together~(\ref{mplower}) and~(\ref{mpupper}).
  The first inequality provides a lower bound of the form
  $|m-p|  \geq \exp(-C_7 s)$ and the second one an upper bound of the
  form $|m-p|  \leq  \exp(C_8 s)|x_1 - x_2|$.
As a result, $|x_1 - x_2| \geq \exp(-(C_7 + C_8) s)$.
 \end{proof}

\section{Final Remarks} \label{final}

In this paper we have obtained a separation bound for the complex roots of trinomials, and we have derived an (easier) algorithmic result:
the number of real roots of a trinomial can be computed in polynomial time.
We conclude with two open problems about trinomials and 4-nomials.
\begin{enumerate}
\item Is it possible to determine in polynomial time the sign of a trinomial $f$ at a rational point? The algorithm should run in time
polynomial in the bit size of the rational point $p/q$ and of the  sparse encoding of $f$ as  defined by~(\ref{size}).

The sign of a binomial at a rational point can be evaluated in polynomial time using Baker's theorem (this is a special case of the comparison problem for succinctly represented integers~\cite{Bastani11,Etessami14}). In time polynomial in the sparse encoding of a polynomial $f \in \zz[X]$ (with any number of monomials) one can also compute its sign  at an integer point~\cite{cucker99}, and one can determine all of its rational roots~\cite{Len99a}.

\item Is it possible to count the number of real roots of a 4-nomial $f$ in time polynomial in the sparse size of $f$? An algorithm for this problem was proposed in~\cite{Bastani11}. The authors perfom an average-case analysis 
  in the arithmetic model as well as in the Turing machine model. Like the present paper, they rely on Baker's theorem.
 Ideally, one would like to have an algorithm with a worst case polynomial running time in the Turing machine model.
\end{enumerate}

\small

\section*{Acknowledgements}
I would like to thank Yann Bugeaud and Michael Sagraloff for their feedback on the first version of this paper. 
In particular, Michael Sagraloff pointed out that in Section~\ref{algosec} the appeal to the root isolation algorithm from his joint paper with Gorav Jindal~\cite{jindal17} needed a more careful justification.
Thanks also go to the anonymous referees for their careful reading
of the manuscript and suggesting additional references.


\end{document}